%% file: main.tex
\documentclass[sigconf,nonacm,authorversion=true]{acmart}

\input{packages}
\input{commands}

\usepackage{balance}

\def\BibTeX{{\rm B\kern-.05em{\sc i\kern-.025em b}\kern-.08emT\kern-.1667em\lower.7ex\hbox{E}\kern-.125emX}}

\copyrightyear{2019}
\acmYear{2019}
\setcopyright{acmcopyright}
\acmConference[PODC '19] {2019 ACM Symposium on Principles of Distributed Computing}{July 29--August 2, 2019}{Toronto, ON, Canada}
\acmBooktitle{2019 ACM Symposium on Principles of Distributed Computing (PODC'19), July 29--August 2, 2019, Toronto, ON, Canada}
\acmPrice{15.00}
\acmDOI{10.1145/3293611.3331589}
\acmISBN{978-1-4503-6217-7/19/07}

\begin{document}

\title{The Consensus Number of a Cryptocurrency}
\titlenote{This is an extended version of a conference article, comprising an additional section (\Cref{sec:generalize-shared}). The conference version of this article appears in the proceedings of the 2019 \emph{ACM Symposium on Principles of Distributed Computing} (PODC’19), July 29–August 2, 2019, Toronto, ON, Canada, \url{https://doi.org/10.1145/3293611.3331589}.}
\subtitle{(Extended Version)}

\author{Rachid Guerraoui}
\email{rachid.guerraoui@epfl.ch}
\affiliation{%
  \institution{EPFL}
  \city{Lausanne}
  \country{Switzerland}
}

\author{Petr Kuznetsov}
\email{petr.kuznetsov@telecom-paristech.fr}
\affiliation{%
  \institution{LTCI, Télécom Paris, IP Paris}
 \city{Paris}
 \country{France}
}

\author{Matteo Monti}
\email{matteo.monti@epfl.ch}
\affiliation{%
  \institution{EPFL}
  \city{Lausanne}
  \country{Switzerland}
}

\author{Matej Pavlovič}
\email{matej.pavlovic@epfl.ch}
\affiliation{%
  \institution{EPFL}
  \city{Lausanne}
 \country{Switzerland}
}

\author{Dragos-Adrian Seredinschi}
\authornote{This work has been supported in part by the European ERC Grant 339539 - AOC.}
\email{dragos-adrian.seredinschi@epfl.ch}
\affiliation{%
  \institution{EPFL}
  \city{Lausanne}
  \country{Switzerland}
}

\input{sections/abstract}

%
%


%

\maketitle

\input{sections/introduction}
\input{sections/model}
\input{sections/shm-cons-1}
\input{sections/shm-cons-k}
\input{sections/implications-in-msg-passing}

\input{sections/msg-passing}
\input{sections/mp-cons-k.tex}

\input{sections/related}

\bibliographystyle{acm}
\bibliography{refs}

%

\end{document}

%% file: packages.tex
\usepackage[utf8]{inputenc}
\usepackage{graphicx}
\usepackage{amsfonts}
\usepackage{amsthm}
\usepackage{xspace}
\usepackage{color}
\usepackage[page]{appendix}

\usepackage{paralist}
\usepackage{wrapfig}

\usepackage{cleveref}
\crefname{section}{Section}{Sections}
\Crefname{section}{Section}{Sections}
\crefname{figure}{Figure}{Figures}
\Crefname{figure}{Figure}{Figures}
\crefname{corollary}{Corollary}{Corollaries}
\Crefname{corollary}{Corollary}{Corollaries}
\crefname{lemma}{Lemma}{Lemmas}
\Crefname{lemma}{Lemma}{Lemmas}

%% file: commands.tex
\newtheorem{lemma}{Lemma}
\newtheorem{theorem}{Theorem}
\newtheorem{corollary}{Corollary}
\newtheorem{definition}{Definition}

\newcommand{\myparagraph}[1]{\paragraph{#1}}
\def\Nat{\mathbb{N}}
\def\to{\rightarrow}

\def\transfer{\textit{transfer}}
\def\read{\textit{read}}

\def\true{\textit{true}}
\def\false{\textit{false}}

\def\lf{\tiny}

\def\nnll{\refstepcounter{linenumber}{\lf\thelinenumber}}

\newcounter{linenumber}

\newcommand{\A}{\mathcal A}

\newcommand{\V}{\mathcal V}

\newcommand{\kcons}{$k$-consensus}

\newcommand{\ignore}[1]{}

\newcommand{\needsrev}[1]{{\color{red}#1}}
\newcommand{\cmt}[2]{\needsrev{[#1]\footnote{\needsrev{#2}}}}
\newcommand{\operation}[2]{$\textit{#1}#2$}

%% file: sections/abstract.tex
\begin{abstract}



  Many blockchain-based algorithms, such as Bitcoin, implement a decentralized \emph{asset transfer system}, often referred to as a \emph{cryptocurrency}.
  As stated in the original paper by Nakamoto, at the heart of these systems lies the problem of preventing
  \emph{double-spending}; this is usually solved by achieving \emph{consensus} on the order of transfers among the participants.  
  In this paper, we treat the asset transfer problem as a
  \emph{concurrent object} and determine its \emph{consensus
    number}, showing that consensus is, in fact, not necessary to prevent double-spending.
  
 We first consider the problem as defined by Nakamoto, where only a single process---the account owner---can withdraw from each account.
Safety and liveness need to be ensured for correct account owners, whereas misbehaving account owners might be unable to perform transfers.
We show that the consensus number of an asset transfer object is $1$.
We then consider  a more general \emph{$k$-shared} asset transfer object where up to $k$ processes can atomically withdraw from the same account, and
show that this object has consensus number $k$.

We establish our results in the context of shared memory with benign faults, allowing us to properly understand the level of difficulty of the asset transfer problem.
We also translate these results in the message passing setting with Byzantine players, a model that is more relevant in practice.
In this model, we describe an asynchronous Byzantine fault-tolerant asset transfer implementation that is both simpler and more efficient than state-of-the-art consensus-based solutions.
Our results are applicable to both the permissioned (private) and permissionless (public) setting, as normally their differentiation is hidden by the abstractions on top of which our algorithms are based.
\end{abstract}

%% file: sections/introduction.tex
\section{Introduction}
\label{sec:intro}

\sloppy
The Bitcoin protocol, introduced in 2008 by Satoshi Nakamoto,  implements a \emph{cryptocurrency}: an electronic decentralized asset transfer system~\cite{nakamotobitcoin}.
Since then, many alternatives to Bitcoin came to prominence.
These include major cryptocurrencies such as Ethereum~\cite{ethereum} or Ripple~\cite{rapoport2014ripple}, as well as systems sparked from research or industry efforts such as Bitcoin-NG~\cite{eyal16bitcoinng}, Algorand~\cite{gilad2017algorand}, ByzCoin~\cite{kogi16byzcoin}, Stellar~\cite{mazieres2015stellar}, Hyperledger~\cite{hyperledger}, Corda~\cite{he16corda}, or Solida~\cite{abr16solida}.
Each alternative brings novel approaches to implementing decentralized transfers, and sometimes offers a more general interface (known as smart contracts~\cite{sz97smartcontr}) than the original protocol proposed by Nakamoto.
They improve over Bitcoin in various aspects, such as performance, energy-efficiency, or security.

A common theme in these protocols, whether they are for basic transfers~\cite{kok18omniledger} or smart contracts~\cite{ethereum},
is that they seek to implement a \emph{blockchain}---a distributed ledger where all the transfers in the system are totally ordered.
Achieving total order among multiple inputs (e.g., transfers) is fundamentally a hard task, equivalent to solving \emph{consensus}~\cite{Her91,HT93}.
Consensus~\cite{FLP85}, a central problem in distributed computing, is known for its notorious difficulty.
It has no deterministic solution in asynchronous systems
if just a single participant can fail~\cite{FLP85}.
Partially synchronous consensus algorithms are tricky to implement correctly~\cite{abrah17revisiting,cac17blwild,clement09making}
and face tough trade-offs between performance, security, and energy-efficiency~\cite{antoni18smr,ber89optimal,gue18blockchain,vuko15quest}.
Not surprisingly, the consensus module is a major bottleneck in blockchain-based protocols~\cite{he16corda,sou18byzantine,vuko15quest}.

A close look at Nakamoto's original paper reveals that the central issue in implementing 
a decentralized asset transfer system (i.e., a cryptocurrency)
is preventing \emph{double-spending}, i.e., spending the same money more than once~\cite{nakamotobitcoin}.
Bitcoin and numerous follow-up systems typically assume that total order---and thus consensus---is vital to preventing double-spending~\cite{gar15backbone}.
There seems to be a common belief, indeed, that a consensus algorithm is essential for implementing decentralized asset transfers~\cite{BonneauMCNKF15,gue18blockchain,kar18vegvisir,nakamotobitcoin}.

As our main result in this paper, we show that this belief is false.
We do so by casting the asset transfer problem as a \emph{sequential object type} and determining that it has \emph{consensus number} $1$ in Herlihy's hierarchy \cite{Her91}.%
\footnote{The consensus number of an object type is the maximal number of processes that can solve consensus using only read-write shared memory and arbitrarily many objects of this type.}

The intuition behind this result is the following.
An asset transfer object maintains a set of accounts.
Each account is associated with an \emph{owner} process that is the only one allowed to issue transfers withdrawing from this account.
Every process can however read the balance of any account.

The main insight here is that relating accounts to unique owners obviates the need for consensus.
It is the owner that decides on the order of transfers from its own account, without the need to agree with any other process---thus the consensus number 1.
Other processes only validate the owner's decisions, ensuring that causal relations across accounts are respected.
We describe a simple asset transfer implementation using  atomic-snapshot memory~\cite{AADGMS93}.
%
A withdrawal from an account is validated by relating the withdrawn amount with the incoming transfers found in the memory snapshot.
Intuitively, as at most one withdrawal can be active on a given account at a time, it is safe to declare the validated operation as successful and post it in the snapshot memory.                

We also present a natural generalization of our result to the setting in which multiple processes are allowed to withdraw from the same account.
A \emph{$k$-shared} asset-transfer object allows up to $k$ processes to execute outgoing transfers from the same account.
We prove that such an object has consensus number $k$ and thus allows for implementing \emph{state machine replication} (now often referred to as \emph{smart contracts}) among the $k$ involved processes using $k$-consensus objects~\cite{JT92}.
%
\ignore{
We prove the lower bound by providing a wait-free implementation of consensus in a shared memory system with $k$ processes equipped with a $k$-shared asset transfer object.
Each process first announces its proposed value in a dedicated (per-process) register,
and then tries to perform a transfer from a shared account.
We choose the initial account balance and the withdrawn amounts so that (1)~exactly one withdrawal wins and (2)~the remaining balance identifies the winner process.
The upper bound proof makes use of a $k$-consensus object~\cite{JT92} (known to have consensus number $k$).
}
We show that $k$-shared asset transfer has consensus number
$k$ by reducing it to $k$-consensus (known to have consensus number $k$) and reducing $k$-consensus to asset transfer.
%

Having established the relative ease of the asset transfer problem using the shared memory model,
we also present a practical solution to this problem in the 
setting of Byzantine fault-prone processes communicating via message passing.
This setting matches realistic deployments of distributed systems.
We describe an asset transfer implementation that does not resort to consensus.
Instead, the implementation relies on a \emph{secure broadcast} primitive that ensures uniform reliable delivery with only weak ordering guarantees~\cite{ma97secure,MR97srm}, circumventing hurdles imposed by consensus.
%
In the $k$-shared case, our results imply that to execute some form of \emph{smart contract} involving $k$ users,
consensus is only needed among these $k$ nodes and not among all nodes in the system.
In particular, should these $k$ nodes be faulty, the rest of the accounts will not be affected.

To summarize, we argue that treating
the asset transfer problem as a concurrent data structure
and measuring its hardness through the lense of distributed computing help understand
it and devise better solutions to it.



The rest of this paper is organized as follows.
We first give the formal definition of the shared memory model and the asset transfer object type~(\cref{sec:model}).
Then, we show that this object type has consensus number $1$~(\cref{sec:avoiding-consensus}).
Next, we generalize our result by proving that a $k$-shared asset transfer object has consensus number $k$ (\cref{sec:k-consensus}).
Finally, we describe the implications of our results in the message passing model with Byzantine faults (\cref{sec:msg-passing-implications,sec:generalize-shared}) and discuss related work (\cref{sec:related-work}).

%% file: sections/model.tex

\section{Shared Memory Model and Asset-Transfer Object Type}
\label{sec:model}

We now present the shared memory model~(\cref{sec:def-shm-model}) and precisely define the problem of \textsf{asset-transfer} as a sequential object type~(\cref{sec:def-sequential}).




\subsection{Definitions}
\label{sec:def-shm-model}

\myparagraph{Processes.}
We assume a set $\Pi$ of $N$ asynchronous processes
that communicate by invoking atomic operations on shared memory objects.
Processes are sequential---we assume that a process never invokes a new operation before obtaining a response from a previous one.

\myparagraph{Object types.}
A sequential object type is defined as a tuple
$T=(Q,q_0,O,R,\Delta)$, where $Q$ is a set of states, $q_0\in Q$ is an
initial state, $O$ is a set of operations, $R$ is a set of responses and
$\Delta\subseteq Q\times\Pi\times O \times Q \times R$ is a relation
that associates a state, a process identifier and an operation to a set of
possible new states and corresponding responses.
We assume that $\Delta$ is total on the first three elements.

A \emph{history} is a sequence of invocations and responses, each
invocation or response associated with a process identifier.
A
\emph{sequential history} is a history that starts with an invocation and in
which every invocation is immediately followed with a response
associated with the same process.
A sequential history
is \emph{legal} if its invocations and responses respect the relation
$\Delta$ for some sequence of state assignments.

\myparagraph{Implementations.}
An \emph{implementation} of an object type $T$ is a distributed algorithm that,
for each process and invoked operation, prescribes the actions that the process needs to
take to perform it.
An \emph{execution} of an implementation is a sequence of
\emph{events}: invocations and responses of operations
or atomic accesses to shared abstractions. The sequence of events at every process
must respect the algorithm assigned to it.

\myparagraph{Failures.}
Processes are subject to \emph{crash} failures
(we consider more general Byzantine failures in the next section).
A process may halt prematurely, in which case we say that the process is \emph{crashed}.
A process is called \emph{faulty} if it crashes during the execution.
A process is \emph{correct} if it is not faulty.
All algorithms we present in the shared memory model are \emph{wait-free}---every correct process eventually returns from each operation it invokes,
regardless of an arbitrary number of other processes crashing or concurrently invoking operations.

\myparagraph{Linearizability.}
For each pattern of operation invocations, the execution produces a
\emph{history}, i.e., a sequence of distinct invocations and responses,
labelled with process identifiers and unique sequence numbers.

A projection of a history $H$ to process $p$, denoted $H|p$ is the
subsequence of elements of $H$ labelled with $p$.
An invocation $o$ by a process $p$ is \emph{incomplete} in $H$ if it is not followed by a
response in $H|p$.
A history is \emph{complete} if it has no incomplete invocations.
A \emph{completion} of $H$ is a history $\bar H$ that is identical to
$H$ except that every incomplete invocation in $H$ is either removed or
\emph{completed} by inserting a matching response somewhere after it.


An invocation $o_1$ \emph{precedes} an invocation $o_2$ in $H$, denoted \hbox{$o_1\prec_H o_2$},
if $o_1$ is complete and the corresponding response $r_1$ precedes $o_2$ in $H$.
Note that $\prec_H$ stipulates a partial order on invocations in $H$.
A \emph{linearizable} implementation  (also said an \emph{atomic
  object}) of  type $T$
ensures that for every history $H$ it produces,
there exists a completion $\bar H$ and a legal sequential history $S$ such that
(1)~for all processes $p$, $\bar H|p=S|p$ and
(2)~$\prec_H\subseteq\prec_S$.

\myparagraph{Consensus number.}
The problem of \emph{consensus} consists for a set of processes to
\emph{propose} values and \emph{decide} on the proposed values
so that no two processes decide on defferent values and every correct
process decides. The \emph{consensus number} of a
type $T$ is the maximal number of processes that can solve consensus
using atomic objects of type $T$ and read-write registers.



\subsection{The asset transfer object type}
\label{sec:def-sequential}

Let $\A$ be a set of \emph{accounts} and $\mu: \A \rightarrow
2^{\Pi}$ be an ``owner'' map that associates each account with a set
of processes that are, intuitively, allowed to debit the account.
We define the \textsf{asset-transfer} object type associated with $\A$ and $\mu$ as a tuple $(Q,q_0,O,R,\Delta)$, where:

\begin{compactitem}

\item The set of states $Q$ is the set of all possible maps
  $q:\;\A\to\Nat$. Intuitively, each state of the object assigns
  each account its \emph{balance}.

\item The initialization map  $q_0:\;\A\to\Nat$ assigns the initial
  balance to each account.

\item Operations and responses of  the type are defined as
  $O=\{\transfer(a,b,x):\; a,b\in\A,\,x\in\Nat\}\cup\{\textit{read}(a):\;a\in\A\}$
  and $R=\{\true,\false\}\cup\Nat$.

 \item $\Delta$ is the set of valid state transitions.
   For a state $q\in Q$, a process $p\in\Pi$, an operation $o\in O$, a response
   $r\in R$ and a new state $q'\in Q$, the tuple $(q,p,o,q',r)\in\Delta$
   if and only if one of the following conditions is satisfied:

  \begin{compactitem}

     \item $o=\transfer(a,b,x) \wedge p\in\mu(a)$ $\wedge$
       $q(a)\geq x$ $\wedge$ $q'(a)=q(a)-x$ $\wedge$ $q'(b)=q(b)+x \wedge \forall c \in{\A}\setminus\{a, b\}: q'(c) = q(c)$ (all other
       accounts unchanged) $\wedge$ $r=\true$;

     \item $o=\transfer(a,b,x)$ $\wedge$  ($p\notin \mu(a)$
       $\vee$ $q(a)< x$) $\wedge$ $q'=q$ $\wedge$ $r=\false$;

     \item $o=\textit{read}(a)$ $\wedge$ $q=q'$ $\wedge$ $r=q(a)$.

   \end{compactitem}

  \end{compactitem}

In other words, operation $\transfer(a,b,x)$ invoked by
process $p$ \emph{succeeds} if and only if $p$ is the owner of the
source account $a$ and
account $a$  has enough balance, and if it does, $x$ is transferred
from $a$ to the destination account $b$.
A $\transfer(a,b,x)$ operation is called \emph{outgoing} for
$a$ and \emph{incoming} for $b$; respectively, the $x$ units are called   \emph{outgoing} for
$a$ and \emph{incoming} for $b$.
A transfer is \emph{successful} if its corresponding response is \emph{true} and \emph{failed} if its corresponding response is \emph{false}.
Operation $\textit{read}(a)$ simply returns the balance of $a$ and
leaves the account balances untouched.

As in Nakamoto's original paper \cite{nakamotobitcoin}, we assume for
the moment that an \textsf{asset-transfer} object has at most one owner per account:
$\forall a \in \A : |\mu(a)| \leq 1$.
Later we lift this assumption and consider more general $k$-shared \textsf{asset-transfer} objects with arbitrary owner maps~$\mu$ (\cref{sec:k-consensus}).
For the sake of simplicity, we also restrict ourselves to transfers with a single source account and a single destination account.
However, our definition (and implementation) of the \textsf{asset-transfer} object type can trivially be extended to support transfers with multiple source accounts (all owned by the same sequential process) and multiple destination accounts.

%% file: sections/shm-cons-1.tex

\section{Asset Transfer Has Consensus Number 1}
\label{sec:avoiding-consensus}

In this section, 
we show that  the \textsf{asset-transfer} type can be wait-free implemented using only read-write registers in a shared memory system with crash failures.
Thus, the type has consensus number~$1$~\cite{Her91}.

Consider an \textsf{asset-transfer} object associated with a set of accounts $\A$ and
an ownership map $\mu$ where $\forall a\in\A$, $|\mu(a)|\leq 1$.
Our implementation is described in Figure~\ref{fig:waitfree}.
Every process $p$ is associated with a distinct location in an
atomic snapshot object~\cite{AADGMS93}
storing the set of all successful {\transfer} operations
executed by $p$ so far.
Since each account is owned by at most one process, all outgoing transfers for an account appear in a single location of the atomic snapshot (associated with the owner process).
This principle bears a similarity to the implementation of a counter object.

Recall that the atomic snapshot (AS) memory is represented as a vector of $N$
shared variables
that can be accessed with two atomic operations: \emph{update} and
\emph{snapshot}. An \emph{update} operation
modifies the value at a given position of the vector and a
\emph{snapshot} returns the state of the whole vector.
We implement the \textit{read} and \textit{transfer} operations as follows.

\begin{figure}[b]
\hrule 
 {\small
\begin{tabbing}
 bbb\=bb\=bb\=bb\=bb\=bb\=bb\=bb \=  \kill
Shared variables: \\
\> $AS$, atomic snapshot, initially $\{\bot\}^N$\\
\\
Local variables: \\
\> $\textit{ops}_p\subseteq \A\times\A\times\Nat$, initially $\emptyset$\\
\\
Upon \transfer$(a,b,x)$\\
\nnll\label{line1:tsf:snapshot}\> $S = AS.\textit{snapshot}()$\\
\nnll\> \textbf{if} $p \notin \mu(a) \vee \textit{balance}(a,S)< x$ \textbf{then}\\
\nnll\label{line1:response-false}\>\> \textbf{return} {\false}\\
\nnll\> $\textit{ops}_p = \textit{ops}_p \cup \{(a,b,x)\}$\\
\nnll\label{line1:tsf:update}\> $AS.\textit{update}(\textit{ops}_p)$\\
\nnll\label{line1:response-true}\>\textbf{return} {\true}\\
\\
Upon \textit{read}$(a)$\\
\nnll\label{line1:read:snapshot}\> $S = AS.\textit{snapshot}()$\\
\nnll\label{line1:read:response}\> \textbf{return} $\textit{balance}(a,S)$\\

\end{tabbing}
 }
 \hrule
\caption{Wait-free implementation of \textsf{asset-transfer}: code for process $p$}
\label{fig:waitfree}
\end{figure}

\begin{compactitem}
\item
To read the balance of an account $a$, the process simply takes a snapshot $S$ and
returns the initial balance plus the sum of incoming amounts minus the sum of all outgoing
amounts. We  denote this number by $\textit{balance}(a,S)$.
As we argue below, the result is guaranteed to be non-negative,
i.e., the operation is correct with respect to the type specification.
\item
To perform $\transfer(a,b,x)$, a process $p$, the owner of $a$, takes a snapshot $S$ and computes $\textit{balance}(a,S)$. 
If the amount to be transferred does not exceed $\textit{balance}(a,S)$,
we add the transfer operation to the set of $p$'s operations in the snapshot object via an \textit{update} operation and return $\true$.
Otherwise, the operation returns $\false$.
\end{compactitem}

\begin{theorem}
  \label{th:waitfree}
  The \textsf{asset-transfer} object type has a wait-free implementation in the
  read-write shared memory model.
\end{theorem}
\begin{proof}
  Fix an execution $E$ of the algorithm in Figure~\ref{fig:waitfree}.
   Atomic snapshots can be wait-free implemented in the read-write
   shared memory model~\cite{AADGMS93}.
  As every operation only involves a finite number of atomic snapshot
  accesses, every process completes each of the operations it invokes
  in a finite number of its own steps.

\noindent Let $\textit{Ops}$ be the set of:
\begin{compactitem}
\item All invocations of \transfer\ or \emph{read} in $E$ that returned, and
\item All invocations of \transfer\ in $E$ that completed the
  \emph{update} operation (line~\ref{line1:tsf:update}). 
\end{compactitem}

Let $H$ be the history of $E$.
We define a completion of $H$ and, for each $o\in\textit{Ops}$, we define a linearization point as follows:

\begin{compactitem}

  \item If $o$ is a {\read} operation, it linearizes at the
    linearization point of the \emph{snapshot} operation in
    line~\ref{line1:read:snapshot}.

  \item   If $o$ is a {\transfer} operation that returns {\false},
    it linearizes at the linearization point of the \emph{snapshot} operation in
    line~\ref{line1:tsf:snapshot}.

 \item   If $o$ is a {\transfer} operation that completed the \emph{update} operation,  it linearizes at the
    linearization point of the \emph{update} operation in
    line~\ref{line1:tsf:update}.
    If $o$ is incomplete in $H$, we complete it with response $\true$.

\end{compactitem}

Let $\bar H$ be the resulting complete history and let $L$ be the sequence
of complete invocations of $\bar H$ in the order of their
linearization points in $E$.
Note that, by the way we linearize invocations, the linearization of a
prefix of $E$ is a prefix of $L$.

Now we show that $L$ is legal and, thus, $H$ is
linearizable.
We proceed by induction, starting with the empty (trivially legal)
prefix of $L$.
Let $L_{\ell}$ be the legal prefix of the first $\ell$ invocations and
$op$ be the $(\ell+1)$st operation of $L$.
Let $op$ be invoked by process $p$.
The following cases are possible:

\begin{compactitem}

\item $op$ is a {\read}$(a)$: the snapshot taken at the linearization point of $op$
  contains all successful transfers concerning $a$ in $L_{\ell}$. By
  the induction hypothesis, the resulting balance is non-negative.

\item $op$ is a failed {\transfer}$(a,b,x)$: the snapshot taken at the linearization point of $op$
  contains all successful transfers concerning $a$ in $L_{\ell}$. By
  the induction hypothesis, the resulting balance is non-negative.

\item  $op$ is a successful  {\transfer}$(a,b,x)$:  by the algorithm,
  before the linearization point of $op$, process $p$ took a snapshot.
  Let $L_{k}$, $k\leq\ell$, be the prefix of $L_{\ell}$ that only
  contain operations linearized before the point in time when the
  snapshot was taken by $p$.

  We observe that $L_k$ includes a \emph{subset} of all incoming transfers on $a$ and
  \emph{all} outgoing transfers on $a$ in $L_{\ell}$. Indeed, as $p$
  is the owner of $a$ and only the owner of $a$ can perform outgoing
  transfers on $a$, all outgoing transfers in $L_{\ell}$ were
  linearized before the moment $p$ took the snapshot within $op$.
  Thus, $\textit{balance}(a,L_k) \leq \textit{balance}(a,L_{\ell})$.%
  \footnote{
    Analogously to $\textit{balance}(a,S)$ that computes the balance for account $a$ based on the transfers contained in snapshot $S$,
    $\textit{balance}(a,L)$, if $L$ is a sequence of operations, computes the balance of account $a$ based on all transfers in $L$.
  }

  By the algorithm, as $op={\transfer}(a,b,x)$ succeeds, we have
  $\textit{balance}(a,L_k)\geq x$.
  Thus, $\textit{balance}(a,L_{\ell})\geq x$ and the resulting balance
  in $L_{\ell+1}$ is non-negative.

\end{compactitem}

Thus, $H$ is linearizable.
\end{proof}

\begin{corollary}\label{cor:consensus}
The \textsf{asset-transfer} object type has \hbox{consensus number~$1$}.
\end{corollary}

%% file: sections/shm-cons-k.tex

\section{\texorpdfstring{$k$}{k}-Shared Asset Transfer Has Consensus Number \texorpdfstring{$k$}{k}}
\label{sec:k-consensus}

We now consider the case with an arbitrary owner map $\mu$.
We show that an \textsf{asset-transfer} object's consensus number is the maximal number of processes sharing an account.
More precisely, the consensus number of an \textsf{asset-transfer} object is $\max_{a\in\A}|\mu(a)|$. 

We say that an \textsf{asset-transfer} object, defined on a set of
accounts $\A$ with an ownership map $\mu$,
is \emph{k-shared} iff $\max_{a\in\A}|\mu(a)|=k$.
In other words, the object is $k$-shared if $\mu$ allows at least one account to be owned by
$k$ processes, and no account is owned by more than $k$ processes.

We show that the consensus number of any $k$-shared \textsf{asset-transfer} object is $k$, which
generalizes our result in \cref{cor:consensus}.
We first show that such an object has consensus number \emph{at least $k$} by implementing consensus for $k$ processes using only registers and an instance of $k$-shared \textsf{asset-transfer}.
We then show that $k$-shared \textsf{asset-transfer} has consensus number \emph{at most $k$} by reducing it to \kcons, an object known to have consensus number $k$~\cite{JT92}.


\begin{lemma}
  \label{lem:cons-to-at}
  Consensus has a wait-free implementation for $k$ processes in the read-write shared memory model equipped with a single $k$-shared \textsf{asset-transfer} object.
\end{lemma}

\begin{figure}
\hrule
 {\small
\setcounter{linenumber}{0}
\begin{tabbing}
bbb\=bb\=bb\=bb\=bb\=bb\=bb\=bb \=  \kill

Shared variables: \\
\> $R[i], i \in 1,\ldots,k$, $k$ registers, initially $R[i] = \bot, \forall i$\\
\> $AT$, $k$-shared \textsf{asset-transfer} object containing:\\
\>\>\> -- an account $a$ with initial balance $2k$\\
\>\>\>\>owned by processes $1,\ldots,k$\\
\>\>\> -- some account $s$\\
\\
Upon \textit{propose}$(v)$:\\
\nnll\label{ln:simple-announce}\> $R[p].write(v)$\\
\nnll\label{ln:simple-transfer}\> $AT.transfer(a, s, 2k - p))$\\
\nnll\label{ln:simple-read}\> \textbf{return} $R[AT.read(a)].read()$

\end{tabbing}
 }
 \hrule
\caption{Wait-free implementation of consensus among $k$ processes using a $k$-shared \textsf{asset-transfer} object. Code for process $p\in\{1,\ldots,k\}$.}
\label{fig:kcons-to-kaccount}
\end{figure}

\begin{proof}

We now provide a wait-free algorithm that solves consensus among $k$ processes using only registers and an instance of $k$-shared \textsf{asset-transfer}.
The algorithm is described in~\cref{fig:kcons-to-kaccount}.
Intuitively, $k$ processes use one shared account $a$ to elect one of them whose input value will be decided.
Before a process $p$ accesses the shared account, $p$ announces its input in a register (line~\ref{ln:simple-announce}).
Process $p$ then tries to perform a transfer from account $a$ to another account.
The amount withdrawn this way from account $a$ is chosen specifically such that:
\begin{compactenum}
\item only one transfer operation can ever succeed, and
\item if the transfer succeeds, the remaining balance on $a$ will uniquely identify process $p$.
\end{compactenum}
To satisfy the above conditions, we initialize the balance of account
$a$ to $2k$ and have each process $p\in\{1,\ldots,k\}$ transfer $2k-p$ (line
\ref{ln:simple-transfer}).
Note that  transfer operations invoked by distinct processes
$p,q\in\{1,\ldots,k\}$ have arguments $2k-p$ and $2k-q$, and
$2k-p+2k-q\geq 2k-k+2k-(k-1)=2k+1$.
The initial balance of $a$ is only $2k$ and no incoming transfers are ever executed.
Therefore, the first transfer operation to be applied to the
object succeeds (no transfer tries to withdraw more then $2k$) and the remaining operations will have to fail due to insufficient balance.

\noindent When $p$ reaches line \ref{ln:simple-read}, at least one transfer must have succeeded:
\begin{compactenum}
\item either $p$'s transfer succeeded, or
\item $p$'s transfer failed due to insufficient balance, in which case
  some other process must have previously succeeded.
\end{compactenum}
Let $q$ be the process whose transfer succeeded. Thus, the balance of account $a$ is $2k - (2k - q) = q$.
Since $q$ performed a transfer operation, by the algorithm, $q$ must
have previously written its proposal to the register $R[q]$.
Regardless of whether $p = q$ or $p \neq q$, reading the balance of account $a$ returns $q$ and $p$ decides the value of $R[q]$.
\end{proof}


To prove that $k$-shared \textsf{asset-transfer} has consensus number at most $k$, we reduce $k$-shared \textsf{asset-transfer} to \kcons.
A \kcons\ object exports a single operation \emph{propose} that, the first $k$ times it is invoked, returns the argument of the first invocation.
All subsequent invocations return $\bot$.
Given that \kcons\ is known to have consensus number exactly
$k$~\cite{JT92}, a wait-free algorithm implementing $k$-shared \textsf{asset-transfer} using
only registers and \kcons\ objects implies that the consensus number of $k$-shared \textsf{asset-transfer} is not more than $k$.

\begin{figure}
\hrule 
 {\small
\setcounter{linenumber}{0}
\begin{tabbing}
bbb\=bb\=bb\=bb\=bb\=bb\=bb\=bb \=  \kill
Shared variables:\\
\> $AS$, atomic snapshot object\\
\> for each $a\in\A$:\\
\>\> $R_a[i], i \in \Pi$, registers, initially $[\bot, \ldots,\bot]$\\
\>\> $kC_a[i], i \geq 0$, list of instances of \kcons\ objects\\
\\
Local variables:\\
\>\> $\textit{hist}$: a set of completed trasfers, initially empty\\
\>for each $a\in \A$:\\
\>\> \textit{committed}$_a$, initially $\emptyset$\\
\>\> \textit{round}$_a$, initially $0$\\
\\
Upon \textit{transfer}$(a, b, x)$:\\
\nnll\> \textbf{if} $p \notin \mu(a)$ \textbf{then}\\
\nnll\>\> \textbf{return} {\false}\\
\nnll\> $tx = (a, b, x, p,\textit{round}_a)$\\
\nnll\label{alg:kwaitfree:register}\> $R_{a}[p].write(tx)$\\
\nnll\> $\textit{collected} = \textit{collect}(a) \setminus \textit{committed}_{a}$\\
\nnll\> \textbf{while} $tx \in \textit{collected}$ \textbf{do} \\
\nnll\label{alg:kwaitfree:picknext}\>\> $\textit{req} = $ the oldest transfer in $\textit{collected}$\\
\nnll\label{alg:kwaitfree:snapshot}\>\> $\textit{prop} = \textit{proposal}(\textit{req}, AS.snapshot())$\\
\nnll\label{alg:kwaitfree:decide}\>\> $\textit{decision} = kC_{a}[\textit{round}_a].propose(prop)$\\
\nnll\>\> $\textit{hist} = \textit{hist} \cup \{\textit{decision}\}$\\
\nnll\label{alg:kwaitfree:update}\>\> $AS.\textit{update}(\textit{hist})$\\
\nnll\>\> $\textit{committed}_a = \textit{committed}_a \cup \{t : \textit{decision}=(t,*)\}$\\
\nnll\>\> $\textit{collected} = \textit{collected} \setminus \textit{committed}_a$\\
\nnll\>\> $\textit{round}_a = \textit{round}_a + 1$\\
\nnll\label{alg:kwaitfree:return}\> \textbf{if} $(tx, \texttt{success}) \in \textit{hist}$ \textbf{then}\\
\nnll\label{alg:kwaitfree:response-true}\>\> \textbf{return} $\true$\\
\nnll\> \textbf{else}\\
\nnll\label{alg:kwaitfree:response-false}\>\> \textbf{return} $\false$\\
\\
Upon \textit{read}$(a)$:\\
\nnll\label{alg:kwaitfree:read-snapshot}\> \textbf{return} \textit{balance}$(a, AS.snapshot())$\\
\\
\textit{collect(a)}:\\
\nnll\> $\textit{collected} = \emptyset$\\
\nnll\> \textbf{for all} $i = \Pi$ \textbf{do}\\
\nnll\>\> \textbf{if} $R_a[i].read() \neq \bot$ \textbf{then}\\
\nnll\>\>\> $\textit{collected} = \textit{collected} \cup \{R_a[i].read()\}$\\
\nnll\> \textbf{return} $collected$\\
\\
\textit{proposal}$((a,b,q,x), \textit{snapshot})$:\\
\nnll\> \textbf{if} $\textit{balance}(a, \textit{snapshot}) \geq x$ \textbf{then}\\
\nnll\>\> $\textit{prop} =  ((a, b, q, x), \texttt{success})$\\
\nnll\> \textbf{else}\\
\nnll\>\> $\textit{prop} =  ((a, b, q, x), \texttt{failure})$\\
\nnll\> \textbf{return} $prop$\\
\\
\textit{balance(a, snapshot)}:\\
\nnll\> $\textit{incoming} = \{tx:  tx = (*, a, *,*,*) \wedge (tx, \texttt{success}) \in \textit{snapshot}\}$\\
\nnll\> $\textit{outgoing} = \{tx:  tx= (a, *, *,*,*) \wedge (tx, \texttt{success}) \in \textit{snapshot}\}$\\
\nnll\> \textbf{return} $q_0(a) + \left(\sum_{(*, a, x,*,*) \in \textit{incoming}} x\right) - \left(\sum_{(a, *, x,*,*) \in \textit{outgoing}} x\right)$
\end{tabbing}
 }
 \hrule
\caption{Wait-free implementation of a {$k$-shared \textsf{asset-transfer}} object using
  \kcons\ objects. Code for process~$p$.}
\label{fig:kaccount-to-kcons}
\end{figure}

The algorithm reducing $k$-shared \textsf{asset-transfer} to \kcons\ is given in~\cref{fig:kaccount-to-kcons}.
Before presenting a formal correctness argument, we first informally explain the intuition of the algorithm.
In our reduction, we associate a series of \kcons\
objects with every account $a$.
Up to $k$ owners of $a$ use the \kcons\
objects to agree on the order of outgoing transfers for $a$.

We maintain the state of the implemented $k$-shared \textsf{asset-transfer} object using an
atomic snapshot object $AS$.
Every process $p$ uses a distinct entry of $AS$ to store a set $\textit{hist}$. $\textit{hist}$ is a subset of all completed outgoing transfers from accounts that $p$ owns (and thus is allowed to debit).
For example, if $p$ is the owner of accounts $d$ and $e$, $p$'s $\textit{hist}$ contains outgoing transfers from $d$ and $e$.
Each element in the $\textit{hist}$ set is represented as $((a, b, x, s, r), \textit{result})$,
where $a, b$, and $x$ are the respective source account, destination account, and the amount transferred, $s$ is the originator of the transfer,
and $r$ is the \emph{round} in which the transfer  was invoked by the originator.
The value of $\textit{result} \in \{\texttt{success}, \texttt{failure}\}$ indicates whether the transfer succeeds or fails.
A transfer becomes ``visible'' when any process inserts it in its corresponding entry of $AS$.

To read the balance of account $a$, a process takes a snapshot of
$AS$, and then sums the initial balance $q_0(a)$ and amounts of all successful incoming transfers,
and subtracts the amounts of successful outgoing transfers found in $AS$.
We say that a successful transfer $tx$ is in a snapshot $AS$ (denoted by $(tx, \texttt{success}) \in AS$) if there exists an entry $e$ in $AS$ such that $(tx, \texttt{success}) \in AS[e]$.

To execute a transfer $o$ outgoing from account $a$, a process $p$ first
announces $o$ in a register $R_a$ that can be written by $p$ and read
by any other process.
This enables a ``helping'' mechanism needed to ensure wait-freedom
to the owners of $a$~\cite{Her91}.

Next, $p$ collects the transfers proposed by other owners and
tries to agree on the order of the collected transfers and their
results using a series of {\kcons} objects.
%
%
For each account, the agreement on the order of transfer-result pairs proceeds in rounds.
Each round is associated with a \kcons\ object which $p$ invokes with a proposal chosen from the set of collected transfers.
Since each process, in each round, only invokes the \kcons\ object
once, no \kcons\ object is invoked more than $k$ times and thus each
invocation returns a value (and not  $\bot$).
%

A transfer-result pair as a proposal for the next
instance of \kcons\ is chosen as follows.
Process $p$ picks the ``oldest'' collected but not yet committed
operation (based on the round number $\textit{round}_a$ attached to
the transfer operation when a process announces it; ties are broken using process IDs).
Then $p$ takes a snapshot of $AS$ and checks whether account $a$ has sufficient balance according to the state represented by the snapshot, and equips the transfer with a corresponding \texttt{success} / \texttt{failure} flag.
The resulting transfer-result pair constitutes $p$'s proposal for the next instance of \kcons.
%
\ignore{
For each round, $p$ appends to its copy of $a$ the decided-upon list of transfer-result pairs for that round.
$p$ keeps executing rounds of \kcons\ until $tx$ has been included in the decision of some round.
This can happen either by $p$'s proposal being decided in some round
or by other processes including $tx$ in their proposals and those proposals being decided (possibly over the course of several rounds).
$p$ locally keeps track of all transfers that have been agreed upon so far and excludes them from the set of collected transfers before constructing its proposal in each round.

Note that the transfer $tx$ announced by $p$ will eventually be included in some decision.
If $p$'s proposal is eventually decided, $tx$ will be decided as part of $p$'s proposal (by the algorithm, $p$ always collects its own transfer).
If $p$ executes alone, $p$'s proposal will eventually be decided in some round, namely the first round for which no other process has invoked the \kcons\ object yet.
The only case where $p$'s proposal is never decided is when at least one other process $q$ successfully invokes \kcons\ infinitely many times, with $q$'s proposal being decided.
In this case, however, $tx$ will eventually also be collected by $q$, and thus become part of $q$'s proposal.
Thus, $p$'s \textit{transfer} operation will always eventually terminate.
}
The currently executed transfer by process $p$ returns as soon as
it is decided by a \kcons\ object, the flag of the decided value (\textit{success}/\textit{failure})
indicating the transfer's response (\true/\false).

\begin{lemma}
  \label{lem:at-to-cons}
  The $k$-shared \textsf{asset-transfer} object type has a wait-free implementation in the
  read-write shared memory model equipped with $k$-consensus objects.
\end{lemma}
\begin{proof}
  We essentially follow the footpath of the proof of Theorem~\ref{th:waitfree}.
  Fix an execution $E$ of the algorithm in
  Figure~\ref{fig:kaccount-to-kcons}.
Let $H$ be the history of $E$.

To perform a transfer $o$ on an account $a$, $p$ \emph{registers} it in $R_a[p]$
(line~\ref{alg:kwaitfree:register}) and then proceeds through a series
of {\kcons} objects, each time collecting $R_{a}$ to learn about the
transfers concurrently proposed by other owners of $a$.
Recall that each {\kcons} object is wait-free.
Suppose, by contradiction, that $o$ is registered in $R_{a}$ but is
never decided by any  instance of {\kcons}.
Eventually, however, $o$ becomes the request with the lowest round
number in $R_{a}$ and, thus, some instance of {\kcons} will be
only accessed with $o$ as a proposed value (line~\ref{alg:kwaitfree:decide}).
By validity of {\kcons}, this instance will return $o$ and, thus, $p$ will be able to complete $o$.

Let $\textit{Ops}$ be the set of all complete operations and
all {\transfer} operations $o$ such that some process
completed the update operation  (line~\ref{alg:kwaitfree:update}) in
$E$ with an argument including $o$
(the atomic snapshot and {\kcons}
operation has been linearized).
Intuitively, we include in $\textit{Ops}$ all operations that
\emph{took effect}, either by returning a response to the user or by
affecting other operations.
Recall that every such {\transfer} operation was
agreed upon in an instance of {\kcons}, let it be  $kC^{o}$.
Therefore, for every such {\transfer} operation $o$, we can identify the
process $q^{o}$ whose proposal has been decided in that instance.
We now determine a completion of $H$ and, for each $o\in\textit{Ops}$, we define a linearization point as
follows:

\begin{compactitem}

  \item If $o$ is a {\read} operation, it linearizes at the
    linearization point of the snapshot operation
    (line~\ref{alg:kwaitfree:read-snapshot}).

 \item   If $o$ is a {\transfer} operation that returns {\false},
    it linearizes at the linearization point of the snapshot operation (line~\ref{alg:kwaitfree:snapshot})
    performed by $q^{o}$ just before it invoked
   $kC^{o}.\textit{propose}()$.

 \item   If $o$ is a {\transfer} operation that some process
included in the update operation (line~\ref{alg:kwaitfree:update}), it linearizes at the
    linearization point of the \emph{first} update operation in $H$
    (line~\ref{alg:kwaitfree:update}) that includes~$o$.
 Furthermore, if $o$ is incomplete in $H$, we complete it with response $\true$.

\end{compactitem}

Let $\bar H$ be the resulting complete history and let $L$ be the sequence
of complete operations of $\bar H$ in the order of their
linearization points in $E$.
%
%
Note that, by the way we linearize operations, the linearization of a
prefix of $E$ is a prefix of $L$.
Also, by construction, the linearization point of an operation
belongs to its interval.

Now we show that $L$ is legal and, thus, $H$ is
linearizable.
We proceed by induction, starting with the empty (trivially legal)
prefix of $L$.
Let $L_{\ell}$ be the legal prefix of the first $\ell$ operation and
$op$ be the $(\ell+1)$st operation of $L$.
Let $op$ be invoked by process $p$.
The following cases are possible:

\begin{compactitem}

\item $op$ is a {\read}$(a)$: the snapshot taken at $op$'s linearization point
  contains all successful transfers concerning $a$ in $L_{\ell}$. By
  the induction hypothesis, the resulting balance is non-negative.
\item $op$ is a failed {\transfer}$(a,b,x)$: the snapshot taken at the linearization point of $op$
  contains all successful transfers concerning $a$ in $L_{\ell}$. By
  the induction hypothesis, the balance corresponding to this snapshot
   non-negative. By the algorithm, the balance is  less than $x$.

\item  $op$ is a successful  {\transfer}$(a,b,x)$.
  Let $L_{s}$, $s\leq\ell$, be the prefix of $L_{\ell}$ that only
  contains operations linearized before the moment of time when $q^{o}$
  has taken the snapshot just before accessing $kC^{o}$.

  As before accessing $kC^{o}$, $q$ went through all preceding
  {\kcons} objects associated with $a$ and put the decided values in
  $AS$,  $L_{s}$ must include \emph{all}  outgoing {\transfer}
  operations for $a$. Furthermore, $L_s$ includes a \emph{subset} of
  all incoming transfers on $a$.
  Thus, $\textit{balance}(a,L_k) \leq \textit{balance}(a,L_{\ell})$.

  By the algorithm, as $op={\transfer}(a,b,x)$ succeeds, we have
  $\textit{balance}(a,L_k)\geq x$.
  Thus, $\textit{balance}(a,L_{\ell})\geq x$ and the resulting balance
  in $L_{\ell+1}$ is non-negative.

\end{compactitem}

Thus, $H$ is linearizable.
\end{proof}


\begin{theorem}
  A $k$-shared \textsf{asset-transfer} object has consensus number $k$.
\end{theorem}

\begin{proof}
  It follows directly from \cref{lem:cons-to-at} that $k$-shared \textsf{asset-transfer} has consensus number at least $k$.
  Moreover, it follows from \cref{lem:at-to-cons} that $k$-shared \textsf{asset-transfer} has consensus number at most $k$.
  Thus, the consensus number of $k$-shared \textsf{asset-transfer} is exactly $k$.
\end{proof}

%% file: sections/implications-in-msg-passing.tex
\section{Asset Transfer in Message Passing}
\label{sec:msg-passing-implications}

We established our theoretical results in a shared memory system with crash failures, proving that consensus is not necessary for implementing an asset transfer system.
Moreover, a natural generalization of such a system where up to $k$ processes have access to atomic operations on the same account has consensus number $k$.
These results help us understand the level of difficulty of certain problems in the domain of cryptocurrencies.
To achieve a practical impact, however, we need an algorithm deployable as a distributed system in a realistic setting.
Arguably, such a setting is one where processes (some of which are potentially malicious) communicate by exchanging messages.

In this section we overview an extension of our results to the message passing system with Byzantine failures. 
Instead of consensus, we rely on a \emph{secure broadcast} primitive that provides reliable delivery with weak (weaker than FIFO) ordering guarantees~\cite{MR97srm}.
Using secure broadcast, processes announce their transfers to the rest of the system.
We establish \emph{dependencies} among these transfers that induce a partial order.
Using a method similar to (a weak form of) vector clocks~\cite{fidgevc}, we make sure that each process applies the transfers respecting this dependency-induced partial order.
In a nutshell, a transfer only depends on all previous transfers outgoing from the same account, and on a subset of transfers incoming to that account.
Each transfer operation corresponds to one invocation of secure broadcast by the corresponding account's owner.
The message being broadcast carries, in addition to the transfer itself, references to the transfer's dependencies.

As secure broadcast only provides liveness if the sender is correct, faulty processes might not be able to perform any transfers.
However, due to secure broadcast's delivery properties, the correct processes will always have a consistent view of the system state.


Every transfer operation only entails a single invocation of secure broadcast and our algorithm does not send any additional messages.
\ignore{
There are different flavors of deterministic and probabilistic secure broadcast implementations presenting various trade-offs ~\cite{toueg-secure,br85acb,ma97secure,MR97srm}.
They claim various complexities ranging from quadratic in system size, through linear in the number of tolerated failures, to ``effectively constant''~\cite{ma97secure}.
}
Our algorithm inherits the complexity from the underlying secure broadcast implementation, and there is plenty of such algorithms optimizing complexity metrics for various settings ~\cite{br85acb,ma97secure,MR97srm,toueg-secure,cachin2001secure,garay2011adaptively,HT93}.
In practice, as shown by a preliminary deployment based on a naive quadratic secure broadcast implementation~\cite{br85acb} in a medium-sized system (up to 100 processes), our solution outperforms a consensus-based one by $1.5x$ to $6x$ in throughput and by up to $2x$ in latency.

The implementation can be further extended to solve the \emph{$k$-shared} asset transfer problem.
As we showed in~\cref{sec:k-consensus}, agreement among a subset of the processes is necessary in such a case.
We associate each account (owned by up to $k$ processes) with a Byzantine-fault tolerant state machine replication (BFT) service executed by the owners~\cite{pbft} of that account.
The BFT service assigns sequence numbers to transfers which the processes then submit to an extended version of the above-mentioned transfer protocol.
As long as the replicated state machine is safe and live, we guarantee that every invoked transfer operation eventually returns.
If an account becomes compromised (i.e., the safety or liveness of the BFT is violated),
only the corresponding account might lose liveness.
In other words, outgoing transfers from the compromised account may not return, while safety and liveness of transfers from ``healthy'' accounts are always guaranteed.
We describe this extension in more details later (\cref{sec:generalize-shared}).

In the rest of this section, we give details on the Byzantine message passing model, adapt our \textsf{asset-transfer} object accordingly (Sec. \ref{sec:def-byz}) and present its broadcast-based implementation (Sec. \ref{sec:at-in-msg-passing}).

%% file: sections/msg-passing.tex
\subsection{Byzantine Message Passing Model}
\label{sec:def-byz}

A process is Byzantine if it deviates from the algorithm it is assigned,
either by halting prematurely, in which case we say that the process is \emph{crashed},
or performing actions that are not prescribed by its algorithm, in which case we say that the process is \emph{malicious}.
Malicious processes can perform arbitrary actions,
except for ones that involve subverting cryptographic primitives (e.g. inverting secure hash functions).
A process is called \emph{faulty} if it is either crashed or malicious. 
A process is \emph{correct} if it is not faulty and \emph{benign} if it is not malicious.
Note that every correct process is benign, but not necessarily vice versa.

We only require that the transfer system behaves correctly towards \emph{benign} processes,
regardless of the behavior of Byzantine ones.
Informally, we want to require that no benign process can be a victim of a double-spending attack,
i.e., every execution appears to benign processes as a correct
sequential execution, respecting the original execution's real-time ordering~\cite{Her91}.

\ignore{
Moreover, we require that such execution respects the real-time order among
operations performed by benign processes~\cite{Her91}: when such an
operation completes, it should be visible to every future operation invoked by a benign process.
}

For the sake of efficiency, in our algorithm, we
slightly relax the last requirement---while still preventing
double-spending.
We require that \emph{successful} \operation{\transfer}{} operations invoked
by benign processes constitute a legal sequential history that preserves the real-time
order.
A \operation{read}{} or a failed \operation{transfer}{} operation invoked
by a benign process $p$ can be
``outdated''---it can be based on a stale state of $p$'s balance.
Informally, one can view the system requirements as \emph{linearizability}~\cite{herl90linearizability}
for successful transfers and \emph{sequential consistency}~\cite{Attiya1994} for failed
transfers and reads.
One can argue that this relaxation incurs little impact on the system's utility, since all incoming transfers are eventually applied.
As progress (liveness) guarantees, we require that every operation
invoked by a correct process eventually completes.

\begin{definition}
\label{def:banking-relaxed}
 Let $E$ be any execution of an implementation and $H$ be
the corresponding history.    
Let $\textit{ops}(H)$ denote the set of operations in $H$ that were
executed by correct processes in $E$.  
An \textsf{asset-transfer} object in message passing guarantees that
each invocation issued by a correct process is followed by a
matching response in $H$, and that there exists $\bar H$,  a completion of $H$, such that: 

\begin{compactenum}

\item[(1)] Let $\bar H^{t}$ denote the sub-history of
  successful transfers of $\bar H$ performed by correct processes and
  $\prec_{\bar H}^t$ be the subset of $\prec_{\bar H}$ restricted to
  operations in $\bar H^{t}$.
  Then there exists a legal sequential history $S$ such that (a)~for every correct process
$p$, $\bar H^{t}|p=S|p$ and (b)~$\prec_{\bar H}^t\subseteq\prec_S$.  

\item[(2)] For every correct process $p$, there exists a legal sequential history $S_p$
  such that:
  
  \begin{compactitem}

  \item $\textit{ops}(\bar H) \subseteq \textit{ops}(S_p)$, and 
    
  \item $S_p|p= \bar H |p$.
    
  \end{compactitem}
     
\end{compactenum}
\end{definition}

Notice that property~(2) implies  that every update in $H$ that
affects the account of a correct process $p$ is eventually included in
$p$'s ``local'' history and, therefore, will reflect
reads and transfer operations subsequently performed by $p$.

\subsection{Asset Transfer Implementation in Message Passing}
\label{sec:at-in-msg-passing}

Instead of consensus, we rely on a secure broadcast primitive
that is strictly weaker than consensus and has a fully asynchronous implementation.
It provides uniform reliable delivery despite Byzantine faults and so-called \emph{source order} among delivered messages.
The source order property, being even weaker than FIFO, guarantees that messages from the same source are delivered in the same order by all correct processes.
More precisely, the secure broadcast primitive we use in our
implementation has the following properties~\cite{MR97srm}:
\begin{itemize}
\item \textbf{Integrity:} A benign process delivers a message $m$ from
  a process $p$ at
  most once and,  if $p$ is benign, only if $p$ previously broadcast $m$.
\item \textbf{Agreement:} If processes $p$ and $q$ are correct and $p$ delivers
  $m$, then $q$ delivers $m$.
\item \textbf{Validity:} If a correct process $p$ broadcasts $m$, then $p$ delivers $m$.
\item \textbf{Source order:} If $p$ and $q$ are benign and both
  deliver $m$ from $r$ and $m'$ from $r$, then they do so in the same order.
\end{itemize}




\myparagraph{Operation.}
%
To perform a transfer $tx$, a process $p$ securely broadcasts a message with the transfer details: the arguments of the \operation{transfer}{} operation (see~\Cref{sec:def-sequential}) and some metadata.
The metadata includes a per-process \emph{sequence number} of $tx$ and references to the \emph{dependencies} of $tx$.
The dependencies are transfers incoming to $p$ that must be known to any process before applying $tx$.
These dependencies impose a causal relation between transfers that must be respected when transfers are being applied.
For example, suppose that process $p$ makes a transfer $tx$ to process $q$ and $q$, after observing $tx$, performs another transfer $tx'$ to process $r$.
$q$'s broadcast message will contain $tx'$, a local sequence number, and a reference to $tx$.
Any process (not only $r$) will only evaluate the validity of $tx'$ after having applied $tx$.
This approach is similar to using vector clocks for implementing causal order among events~\cite{fidgevc}.

To ensure the authenticity of operations---so that no process
is able to debit another process's account---we assume that
processes sign all their messages before broadcasting them.
In practice, similar to Bitcoin and other transfer systems,
every process possesses a public-private key pair that allows only $p$ to securely initiate transfers from its corresponding account.
For simplicity of presentation, we omit this mechanism in the algorithm pseudocode.

\begin{figure}
\hrule \vspace{1mm}
 {\small
\setcounter{linenumber}{0}
\begin{tabbing}
 bbb\=bb\=bb\=bb\=bb\=bb\=bb\=bb \=  \kill
\textbf{Local variables}: \\
$\textit{seq}[\textit{ }]$, initially $\textit{seq}[q] = 0$, $\forall q$ \emph{\{Number of validated transfers outgoing from $q$\}}\\
$\textit{rec}[\textit{ }]$, initially $\textit{rec}[q] = 0$, $\forall
q$ \emph{\{Number of delivered transfers from $q$\}}
\\
$hist[\textit{ }]$, initially $hist[q] = \emptyset$, $\forall q$ \emph{\{Set of validated transfers involving $q$\}}
\\
$\textit{deps}$, initially $\emptyset$ \emph{\{Set of last incoming transfers for account of local process p\}}
\\
$\textit{toValidate}$, initially $\emptyset$ \emph{\{Set of delivered (but not validated) transfers\}}
\\
\\
\nnll{}{} \textbf{operation} \textit{transfer}$(a,b,x)$ where $\mu(a) = \{p\}$\\
\nnll\label{line:check-balance}\> \textbf{if} $\textit{balance}(a,hist[p]\cup\textit{deps})< x$
\textbf{then}\\
\nnll\label{line:response-false}\>\> \textbf{return} \textit{false}\\
\nnll\label{line:check-sbcast}\>
$\textit{broadcast}([(a,b,x,\textit{seq}[p]+1),\textit{deps}])$\\ 
\nnll\label{line:dep-null}\> $\textit{deps} =\emptyset$\\
\\
\nnll{}{} \textbf{operation} \textit{read}$(a)$\\
\nnll\label{line:response-read}\> \textbf{return} $\textit{balance}(a,hist[a]\cup\textit{deps})$\\
\\
\nnll{}{}\label{line:deliver} \textbf{upon} $\textit{deliver}(q,m)$\\
\nnll\label{line:sb-deliver}\> let $m$ be $[(q,d,y,s),h]$\\
\nnll\label{line:nextrec}\> \textbf{if} $s=\textit{rec}[q]+1$ \textbf{then} \\
\nnll\>\> $\textit{rec}[q]=\textit{rec}[q]+1$\\
\nnll\>\> $\textit{toValidate} =  \textit{toValidate}\cup\{(q,m)\}$\\

\\
\nnll\label{line:validate}{}{} \textbf{upon} $(q,[t,h])\in \textit{toValidate} \;\wedge \;\textit{Valid}(q,t,h)$\\
\nnll\> let $t$ be $(q,d,y,s)$\\

\nnll\label{line:append-outgoing}\> $hist[q] := hist[q] \cup h \cup \{t\}$\\


\nnll\>\> $\textit{seq}[q] = s$\\
\nnll\>\> \textbf{if} $d=p$ \textbf{then} \\
\nnll\label{line:deps}\>\>\> $\textit{deps} = \textit{deps}\cup (q,d,y,s)$
\\
\nnll\>\> \textbf{if} $q=p$ \textbf{then} \\
\nnll\label{line:response-true}\>\>\> \textbf{return} \textit{true}\\
\\
\nnll\label{line:valid}{}{} \textbf{function} $\textit{Valid}(q,t,h)$\\
\nnll\> let $t$ be $(c,d,y,s)$\\
\nnll\label{line:validation}\> \textbf{return} ($q=c$)\\
\nnll\label{line:check-hist1}\>\>\>\>\textbf{and} ($s= \textit{seq}[q]+1$)\\
\nnll\label{line:check-hist2}\>\>\>\>\textbf{and} ($\textit{balance}(c,hist[q])\geq y$)\\

\nnll\label{line:check-hist3}\>\>\>\>\textbf{and} ($\forall(a,b,x,r) \in h: (a,b,x,r) \in hist[a]$)\\


\\
\nnll{}{} \textbf{function} $\textit{balance}(a,h)$\\
\nnll\label{line:balance}\> \textbf{return} sum of incoming minus outgoing transfers for account $a$ in $h$
\end{tabbing}
 }
 \hrule
\caption{Consensusless transfer system based on secure broadcast. Code for every process $p$.}

\label{fig:banking-relaxed}
\vspace{-.5cm}
\end{figure}



\Cref{fig:banking-relaxed} describes the full algorithm implementing \textsf{asset-transfer} in a Byzantine-prone message passing system.
Each process $p$ maintains, for each process $q$, an integer
$\textit{seq}[q]$ reflecting the number of transfers which
process $q$ initiated and which process $p$ has validated and applied.
Process $p$ also maintains, for every process $q$, an integer
$\textit{rec}[q]$ reflecting the number of transfers process $q$ has initiated and process $p$ has delivered (but not necessarily applied).

Additionally, there is also a list $\textit{hist}[q]$ of transfers which \emph{involve} process $q$.
We say that a transfer operation involves a process $q$ if that transfer is either outgoing or incoming on the account of $q$.
Each process $p$ maintains as well a local variable $\textit{deps}$.
This is a set of transfers \emph{incoming} for $p$ that $p$ has applied since  the last successful \emph{outgoing} transfer.
Finally, the set $\textit{toValidate}$ contains delivered transfers that are pending validation (i.e., have been delivered, but not yet validated).

To perform a $\textit{transfer}$ operation, process $p$ first checks the
balance of its own account, and if the balance is insufficient, returns $\textit{false}$ (line~\ref{line:response-false}).
Otherwise, process $p$ broadcasts a message with this operation via the secure broadcast primitive (line~\ref{line:check-sbcast}).
This message includes the three basic arguments of a \operation{transfer}{} operation as well as $\textit{seq}[p]+1$ and dependencies $\textit{deps}$.
Each correct process in the system eventually delivers this message via secure broadcast (line~\ref{line:deliver}).
Note that, given the assumption of no process executing more than one concurrent transfer, every process waits for delivery of its own message before initiating another broadcast.
This effectively turns the source order property of secure broadcast into FIFO order.
Upon delivery, process $p$ checks this message for well-formedness (lines~\ref{line:sb-deliver} and~\ref{line:nextrec}), and then adds it to the set of messages pending validation.
We explain the validation procedure later.

Once a transfer passes validation (the predicate in line~\ref{line:validate} is satisfied), process $p$ applies this transfer on the local state.
Applying a transfer means that process $p$ adds this transfer and its dependencies to the history of the outgoing (line~\ref{line:append-outgoing}) account.
If the transfer is incoming for local process $p$, it is also added to $\textit{deps}$, the set of current dependencies for $p$ (line~\ref{line:deps}).
If the transfer is outgoing for $p$, i.e., it is the currently pending
\operation{transfer}{} operation invoked by $p$, then the response $\textit{true}$ is returned (line~\ref{line:response-true}).

To perform a \textit{read}$(a)$ operation for account $a$, process $p$ simply computes the
balance of this account based on the local history $hist[a]$ (line~\ref{line:balance}).

Before applying a transfer $op$ from some process $q$, process $p$ validates $op$ via the \emph{Valid} function (lines~\ref{line:valid}--\ref{line:check-hist3}).
To be valid, $op$ must satisfy four conditions.
The first condition is that process $q$ (the issuer of transfer $op$) must be the owner of the outgoing account for $op$ (line~\ref{line:validation}).
Second, any preceding transfers that process $q$ issued must have been validated (line~\ref{line:check-hist1}).
Third, the balance of account $q$ must not drop below zero (line~\ref{line:check-hist2}).
Finally, the reported dependencies of $op$ (encoded in $h$ of
line~\ref{line:check-hist3}) must have been validated and exist in
$\textit{hist}[q]$.

\begin{lemma}
  \label{lem:liveness}
In any infinite execution of the algorithm (Figure~\ref{fig:banking-relaxed}),
every operation performed by a correct process eventually completes.
\end{lemma}
\begin{proof}
A transfer operation that fails or a read operation invoked by a correct process
returns immediately (lines~\ref{line:response-false}
and~\ref{line:response-read}, respectively).

Consider a transfer operation $T$ invoked by a correct process $p$
that \emph{succeeds} (i.e., passes
the check in line~\ref{line:check-balance}), so $p$ broadcasts a message with the transfer details using secure broadcast (line~\ref{line:check-sbcast}).
%
By the validity property of secure broadcast,
$p$ eventually delivers the message (via the secure broadcast callback, line~\ref{line:deliver}) and adds it to the
$\textit{toValidate}$ set.
%
By the algorithm, this message includes a set $\textit{deps}$ of
operations (called $h$, line~\ref{line:sb-deliver}) that involve $p$'s account.
This set includes transfers that process $p$ delivered and validated after issuing
the prior successful outgoing transfer (or since system initialization if there is no such transfer)
but before issuing $T$ \hbox{(lines~\ref{line:check-sbcast} and~\ref{line:dep-null})}.

As process $p$ is correct, it operates on its own account, respects the
sequence numbers, and issues a transfer only if it has enough balance
on the account. Thus, when it is delivered by $p$, $T$ must satisfy the first three conditions of the
$\textit{Valid}$ predicate (lines~\ref{line:validation}--\ref{line:check-hist2}).
Moreover, by the algorithm, all dependencies (labeled $h$ in function $\textit{Valid}$) included in $T$  are in the history
$hist[p]$ and, thus the fourth validation condition (line~\ref{line:check-hist3})
also holds.

Thus, $p$ eventually validates $T$ and completes the operation by
returning $\true$ in line~\ref{line:response-true}.
\end{proof}

\begin{theorem}
The algorithm in Figure~\ref{fig:banking-relaxed} implements an
\textsf{asset-transfer} object type.
\end{theorem}
\begin{proof}
Fix an execution $E$ of the algorithm, let $H$ be the
  corresponding  history.

Let $\V$ denote the set of all messages 
  that were delivered (line~\ref{line:deliver}) and
validated (line~\ref{line:validation}) at correct processes in $E$.
Every message $m=[(q,d,y,s),h]\in\V$ is put
in $hist[q]$ (line~\ref{line:append-outgoing}). 
%
We define an order $\preceq \subseteq \V\times\V$ as
follows. For $m=[(q,d,y,s),h]\in \V$ and $m'=[(r,d',y',s'),h']\in \V$,
we have $m\preceq m'$ if and only if one of the following conditions holds:

\begin{itemize}
\item $q=r$ and $s<s'$,

\item $(r,d',y',s')\in h$, or

\item there exists $m''\in \V$ such that $m\preceq m''$ and
  $m''\preceq m'$.
\end{itemize}

By the source order property of secure broadcast (see~\Cref{sec:at-in-msg-passing}), correct processes
$p$ and $r$ deliver messages from any process $q$ in the same order.
By the algorithm in~\Cref{fig:banking-relaxed}, a message from $q$ with a sequence number $i$ is
added by a correct process to
$\textit{toValidate}$ set
only if the previous message from $q$ added to $\textit{toValidate}$ had
sequence number $i-1$ (line~\ref{line:nextrec}).
Furthermore, a message  $m=[(q,d,y,s),h]$ is validated at a correct
process only if all messages in $h$ have been previously validated (line~\ref{line:check-hist3}).
Therefore, $\preceq$ is acyclic and
thus can be extended to a total order.

Let $S$ be the sequential history constructed from any such total order on messages
in $\V$ in which every message $m=[(q,d,y,s),h]$ is replaced with the
invocation-response pair
$\textit{transfer}(q,d,y);\true$.

By construction, every operation $\textit{transfer}(q,d,y)$ in $S$ is preceded by a sequence of
transfers that ensure that the balance of $q$ does not drop below $y$
(line~\ref{line:check-hist2}).
In particular, $S$ includes all outgoing transfers from the account of
$q$ performed previously by $q$ itself.
Additionally $S$ may order some \emph{incoming} transfer to $q$
that did not appear at $hist[q]$ before the corresponding $(q,d,y,s)$
has been added to it.
But these ``unaccounted'' operations may only increase the balance
of $q$ and, thus, it is indeed legal
to return $\true$.

By construction, for each correct process $p$, $S$ respects the order of
successful transfers issued by $p$.
Thus, the subsequence of successful transfers in $H$ ``looks'' linearizable to
the correct processes:  $H$, restricted to successful transfers
witnessed by the correct processes, is consistent with a legal sequential history $S$.

Let $p$ be a correct process in $E$.
Now let $\V_p$ denote the set  of all messages 
  that were delivered (line~\ref{line:deliver}) and
validated (line~\ref{line:validation}) at $p$  in $E$.
Let $\preceq_p$ be the subset of $\preceq$ restricted to the elements in
$\V_p$.
Obviously, $\preceq_p$ is cycle-free and we can again extend it to a
total order.
Let $S_p$ be the sequential history build in the same way as $S$
above.
Similarly, we can see that $S_p$ is legal and, by construction,
consistent with the local history of \emph{all} operations of $p$
(including reads and failed transfers).

By Lemma~\ref{lem:liveness}, every operation invoked by a correct
process eventually completes.
Thus, $E$ indeed satisfies the properties of an \textsf{asset-transfer} object type.
\end{proof}

%% file: sections/mp-cons-k.tex

\section{\texorpdfstring{$k$}{k}-shared Asset Transfer in Message Passing}
\label{sec:generalize-shared}

Our message-passing \textsf{asset-transfer} implementation can be naturally extended to the $k$-shared case,  when some
accounts are owned by up to $k$ processes.
As we showed in~\Cref{sec:k-consensus}, a purely asynchronous implementation of a $k$-shared
\textsf{asset-transfer} does not exist, even in the benign shared-memory environment.

\myparagraph{\texorpdfstring{$k$}{k}-shared BFT service.}
To circumvent this impossibility, we assume that every account is
associated with a Byzantine fault-tolerant state-machine replication
service (BFT~\cite{pbft}) that is used by the account's owners to
order their outgoing transfers.
More precisely, the transfers issued by the owners are assigned
monotonically increasing \emph{sequence numbers}.

The service can be implemented by the owners themselves, acting
both as \emph{clients}, submitting requests, and \emph{replicas},
reaching agreement on the order in which the requests must be served.
As long as more than two thirds of the owners are correct, the service
is \emph{safe}, in particular, no sequence number is assigned to more
than one transfer.
%
Moreover, under the condition that the owners can eventually communicate within
a bounded message delay, every request submitted by a correct owner is
guaranteed to be eventually assigned a sequence number~\cite{pbft}.
One can argue that it is much more likely that this assumption of \emph{eventual
  synchrony} holds for a bounded set of owners, rather than for the
whole set of system participants.
Furthermore, communication complexity of such an implementation is
polynomial in $k$ and not in $N$, the number of processes.

\myparagraph{Account order in secure broadcast.}
Consider even the case where the threshold of one third of
Byzantine owners is exceeded, where the account may become
blocked or, even worse, compromised.
In this case, different owners may be able to issue two different
transfers associated with the same sequence number.

This issue can be mitigated by a slight modification of the classical
secure broadcast algorithm~\cite{MR97srm}.
In addition to the properties of Integrity, Validity and Agreement of
secure broadcast,  the modified algorithm can implement the property of
\emph{account order}, generalizing the \emph{source order} property (\Cref{sec:at-in-msg-passing}).
Assume that each broadcast message is equipped with a sequence
number (generated by the BFT service, as we will see below).
\begin{itemize}
\item \textbf{Account order:} If a benign process $p$
  delivers messages $m$ (with sequence number $s$)  and $m'$ (with
  sequence number $s'$) such that $m$ and $m'$ are associated with the
  same account and $s<s'$, then $p$ delivers $m$ before $m'$.
\end{itemize}

Informally, the implementation works as follows. The sender sends the
message (containing the account reference and the sequence number) it
wants to broadcast to all and waits until it receives acknowledgements
from a \emph{quorum} of more than two thirds of the processes.
A message with a sequence number $s$ associated with an account $a$ is
only acknowledged by a benign process if the last message
associated with $a$ it delivered had sequence number $s-1$.
Once a quorum is collected, the sender sends the message equipped with
the signed quorum to all and delivers the message.
This way, the benign processes deliver the messages associated with
the same account in the same order.
If the owners of an account send conflicting messages for the same
sequence number, the account may block.
However, and most importantly, even a compromised account is always prevented from double spending.
Liveness of operations on a compromised account is not guaranteed, but safety and liveness of other operations remains unaffected.

\myparagraph{Putting it all together.}
The resulting $k$-shared \textsf{asset transfer} system is a
composition of a collection of BFT services (one per account), the
modified secure broadcast protocol (providing the account-order
property), and a slightly modified protocol in
Figure~\ref{fig:banking-relaxed}.

To issue a transfer operation $t$ on an account $a$ it owns, a process $p$ first
submits $t$ to the associated BFT service to get a sequence
number. Assuming that the account is not compromised and the service is
consistent, the transfer receives a unique sequence number $s$.
Note that the decided tuple $(a,t,s)$ should be signed by a quorum of
owners: this will be used by the other processes in the system to
ensure that the sequence number has been indeed agreed upon by the
owners of $a$.
The process executes the protocol in Figure~\ref{fig:banking-relaxed},
with the only modification that the sequence number \textit{seq} is
now not computed locally but adopted from the BFT service.

Intuitively, as the transfers associated with a given account are
processed by the benign processes in the same order,  the resulting protocol ensures that the history of
successful transfers is linearizable.
On the liveness side, the protocol ensures that every transfer on a
non-compromised account is guaranteed to complete.

%% file: sections/related.tex
\section{Related Work}
\label{sec:related-work}



Many systems address the problem of asset transfers, be they for a permissioned (private, with a trusted external access control mechanism)~\cite{hyperledger,he16corda,kar18vegvisir} or permissionless (public, prone to Sybil attacks) setting~\cite{abr16solida,decker2016bitcoin,gilad2017algorand,kok18omniledger,nakamotobitcoin,rocket}.
Decentralized systems for the public setting are open to the world.
To prevent malicious parties from overtaking the system, these systems rely on Sybil-proof techniques, e.g., proof-of-work~\cite{nakamotobitcoin}, or proof-of-stake~\cite{be16cryptocurrencies}.
The above-mentioned solutions, whether for the permissionless or the permissioned environment, seek to solve consensus.
They must inevitably rely on synchrony assumptions or randomization.
By sidestepping consensus, we can provide a deterministic and asynchronous implementation.

It is worth noting that many of those solutions allow for more than just transfers, and support richer operations on the system state---so-called smart contracts.
Our paper focuses on the original asset transfer problem, as defined by Nakamoto~\cite{nakamotobitcoin}, and we do not address smart contracts, for certain forms of which consensus is indeed necessary.
However, our approach allows for arbitrary operations, if those operations affect groups of the participants that can solve consensus among themselves.
Potential safety or liveness violations of those operations (in case this group gets compromised) are confined to the group and do not affect the rest of the system.

In the blockchain ecosystem, a lot of work has been devoted to avoid  a totally ordered chain of transfers.
The idea is to replace the totally ordered linear structure of a blockchain with that of a directed acyclic graph (DAG) for structuring the transfers in the system.
Notable systems in this spirit include Byteball~\cite{chu6byteball}, Vegvisir~\cite{kar18vegvisir}, Corda~\cite{he16corda}, Nano~\cite{lemahieu2018nano}, or the GHOST protocol~\cite{som13accelerating}.
Even if these systems use a DAG to replace the classic blockchain, they still employ consensus.

We can also use a DAG to characterize the relation between transfers, but we do not resort to solving consensus to build the DAG, nor do we use the DAG to solve consensus.
More precisely, we can regard each account as having an individual history.
Each such history is managed by the corresponding account owner without depending on a global view of the system.
Histories are loosely coupled through a causality relation established by dependencies among transfers.

The important insight that an asynchronous broadcast-style abstraction suffices for transfers appears in the literature as early as 2002, due to Pedone and Schiper~\cite{ped02handling}.
Duan et. al.~\cite{du18beat} introduce efficient Byzantine fault-tolerant protocols for storage and also build on this insight.
So does recent work by Gupta~\cite{gup16nonconsensus} on financial transfers which seems closest to ours; the proposed algorithm is based on similar principles as some implementations of secure broadcast~\cite{ma97secure,MR97srm}.
To the best of our knowledge, however, we are the first to formally define the asset transfer problem as a shared object type, study its consensus number, and propose algorithms building on top of standard abstractions that are amenable to a real deployment in cryptocurrencies.